\documentclass{article}

\usepackage[preprint]{neurips_2026}

\usepackage[utf8]{inputenc} %
\usepackage[T1]{fontenc}    %
\usepackage{hyperref}       %
\usepackage{doi}
\usepackage{url}            %
\usepackage{booktabs}       %
\usepackage{array}          %
\usepackage{amsfonts}       %
\usepackage{nicefrac}       %
\usepackage{microtype}      %
\usepackage{xcolor}         %
\usepackage{tikz}           %
\usetikzlibrary{positioning,arrows.meta,calc}
\usepackage{adjustbox}      %
\usepackage{float}          %

\hypersetup{
    colorlinks=true,       %
    linkcolor=blue!70!black, %
    citecolor=green!40!black, %
    urlcolor=blue!80!black    %
}

\title{Proving at Scale for Universal Algebra}

\author{%
  Jo\~{a}o Ara\'{u}jo \\
  Universidade Nova de Lisboa, Portugal \\
  \And
  Jan H\r{u}la \\
  Czech Technical University in Prague, Czechia \\
  \And
  Mikol\'{a}\v{s} Janota \\
  Czech Technical University in Prague, Czechia \\
  \And
  Edmond W.\ H.\ Lee \\
  Nova Southeastern University, USA\\
  \And
  Bartosz Naskr\k{e}cki \\
  Adam Mickiewicz University, Pozna\'{n}, Poland \\
  Warsaw University of Technology, Warsaw, Poland \\
}

\usepackage{amsmath,amssymb,amsthm}
\theoremstyle{definition}
\newtheorem{example}{Example}
\theoremstyle{plain}
\newtheorem{proposition}{Proposition}
\usepackage{xspace}
\newcommand{\project}{\textsc{SemiBase}\xspace}

\newcommand{\nSixRaw}{15973}        %
\newcommand{\nSealedSixRaw}{15973}  %
\newcommand{\nLowOrder}{\num{1309}\xspace}   %
\newcommand{\nSix}{\num{\nSixRaw}\xspace}
\newcommand{\nSealedSix}{\num{\nSealedSixRaw}\xspace}
\newcommand{\nOpenSix}{\num{\the\numexpr\nSixRaw-\nSealedSixRaw\relax}\xspace} %

\usepackage[disable]{todonotes}

\usepackage{tikz,graphicx}
\usepackage{subcaption}     %

\makeatletter\@secpenalty=0\makeatother
\usetikzlibrary{arrows.meta,calc,positioning}
\definecolor{apInk}{HTML}{183D5D}
\definecolor{apMuted}{HTML}{536574}
\definecolor{apLine}{HTML}{9EADB7}
\definecolor{apPanel}{HTML}{F6F8FA}
\definecolor{apTeal}{HTML}{227D8D}
\definecolor{apTealLight}{HTML}{EDF6F6}
\definecolor{apBlue}{HTML}{42759A}
\definecolor{apBlueLight}{HTML}{EFF4F9}
\definecolor{apAmber}{HTML}{A57424}
\definecolor{apAmberLight}{HTML}{FFF8EB}
\definecolor{apGreen}{HTML}{397B59}
\definecolor{apGreenLight}{HTML}{EFF7F1}
\definecolor{apRed}{HTML}{A75446}
\tikzset{
  ap/base/.style={font=\sffamily\fontsize{8.7}{10.7}\selectfont,
    text=apInk,inner sep=0pt,outer sep=0pt},
  ap/panel/.style={draw=apLine!65,fill=apPanel,rounded corners=3pt,line width=.5pt},
  ap/box/.style={draw=apLine,fill=white,rounded corners=2pt,line width=.65pt,
    align=center,inner xsep=5pt,inner ysep=6pt},
  ap/fable/.style={ap/box,draw=apTeal,fill=apTealLight},
  ap/codex/.style={ap/box,draw=apBlue,fill=apBlueLight},
  ap/mail/.style={ap/box,draw=apAmber,fill=apAmberLight},
  ap/accepted/.style={ap/box,draw=apGreen,fill=apGreenLight},
  ap/work/.style={<->,>={Latex[length=1.7mm,width=1.2mm]},
    draw=apMuted!80,line width=.7pt,rounded corners=2pt},
  ap/flow/.style={-{Latex[length=1.9mm,width=1.35mm]},
    draw=apInk,line width=.85pt,rounded corners=2pt},
  ap/fablemail/.style={ap/flow,draw=apTeal,line width=1.15pt},
  ap/codexmail/.style={ap/flow,draw=apBlue,line width=1.15pt},
  ap/failure/.style={ap/flow,draw=apRed,dashed,line width=.7pt},
  ap/tiny/.style={font=\sffamily\fontsize{7.4}{9.3}\selectfont,text=apMuted},
  ap/section/.style={font=\sffamily\bfseries\fontsize{9.4}{11.2}\selectfont,text=apInk},
  ap/code/.style={font=\ttfamily\fontsize{8.3}{10.4}\selectfont,
    text=apInk,align=left,anchor=north west},
  ap/edgelabel/.style={font=\sffamily\fontsize{7.1}{8.4}\selectfont,
    fill=apPanel,inner sep=1pt,text=apMuted}
}
 
\usepackage{siunitx}
\usepackage{listings}
\definecolor{LeanKeyword}{HTML}{1F4E79}
\definecolor{LeanTactic}{HTML}{6A3D9A}
\definecolor{LeanType}{HTML}{006D77}
\definecolor{LeanComment}{HTML}{5C6F5C}
\definecolor{LeanString}{HTML}{A23E48}
\definecolor{LeanRule}{HTML}{9AAABD}
\definecolor{LeanBackground}{HTML}{F7F9FC}
\lstdefinelanguage{lean4}{
  morekeywords={def,theorem,abbrev,inductive,structure,where,namespace,end,
    open,import,private,match,with,let,have,show,if,then,else,deriving},
  morekeywords=[2]{by,fun,cases,induction,intro,exact,refine,apply,simp,
    simpa,rw,subst,rfl,decide,only,at,using,change},
  morekeywords=[3]{Prop,Type,Nat,Fin,List,Word,Identity,Semigroup,
    FiniteTable,Derives},
  sensitive=true,
  morecomment=[l]{--},
  morecomment=[s]{/-}{-/},
  morestring=[b]",
}
\begin{document}

\maketitle

\begin{abstract}
We introduce \project, a project that computes and formally certifies finite
identity bases for small semigroups. Deciding finite basability is undecidable
for finite algebras and remains open for finite semigroups. The task requires a
proof that a candidate basis is complete, or a proof that none exists, rather
than a single first-order validity query. LLM-guided agents search for these
proofs; a referee agent rebuilds them from source, and the Lean kernel checks
the resulting corpus in a final audit. Humans choose targets and approve final
outcomes. We certify every semigroup of order at most 6: all \nLowOrder
semigroups of order at most 5 and all \nSealedSix of order 6, including proofs
that the four known nonfinitely based semigroups have no finite basis. The
bases for order 6 define \num{505} distinct varieties, whose complete inclusion
order we determine by finite evaluation. The resulting catalogue is a
machine-checked account of results scattered across the literature and a
tested foundation for order 7.
\end{abstract}

\section{Introduction}

\citet{tarski1968} asks the following question about a given finite algebra: can
all identities that hold in it be described by a finite set of identities?
Consider, for example, the group $\mathbb{Z}_2$ with addition modulo~2. In
addition to the group axioms, it satisfies the
identity $x+x=0$. It also satisfies $x+x+x+x=0$, but this is a simple
consequence of the former; in fact, every identity of the group follows from
it. The group $\mathbb{Z}_2$ therefore has a finite basis consisting of a single
identity, assuming the group axioms. \citet{Oates1964} show that
finite groups always have a finite basis. \citet{undecidability} shows that given a finite algebra as
input, deciding whether it has a finite basis is \emph{undecidable}.

The project \project tackles the finite-basis problem for finite semigroups, a setting that
has been studied extensively; see the surveys of \cite{volkov01} and \citet{survey}. The smallest
nonfinitely based semigroups have 6 elements; in fact, there are only 4 distinct
nonfinitely based semigroups of that order as shown by \citet{Lee2012}. A complete catalogue of nonfinitely
based semigroups of order 7 is not yet known. In general, it is not known
whether the existence of a finite basis is decidable, but it is believed to be
undecidable. This provides a setting for our project to explore the frontiers of
what is known about small semigroups.

The project is building a catalogue in which each covered semigroup is equipped
with either a certified finite basis or a certified proof that it does not have one.
Proofs are carried out in Lean~4~\citep{lean}. Currently, the project covers
all \nLowOrder semigroups of order at most 5 and all \nSix semigroups of
order 6. The work was done by language-model agents. At our direction they
wrote the scripted pipeline that settled most of the catalogue; they proposed
candidate bases and searched for counter-models to them, wrote the proofs shared
by families of semigroups and the generators of the individual proofs, and ran
every build. Humans chose which semigroups to attack, approved changes to the
repository and restarted stalled agents, with about \num{1100} prompts over
three months (\autoref{sec:effort}). Mathematicians working on the finite
basis problem stay in the loop~\citep{avigad2026mathematiciansageai}.

To our knowledge, \project is unique in using an autonomous, agent-driven
workflow to pursue a systematic classification of finite identity bases for
small semigroups. A useful comparison is the Equational Theories Project (ETP),
a large-scale, machine-assisted study of equational reasoning whose results are
checked in Lean. ETP determined the implication relations among
laws on magmas by combining human-generated proofs and constructions
with automated theorem provers and other computational tools~\citep{etp}.
Subsequent experiments showed that Vampire proves every valid implication
in this benchmark and refutes the vast majority of the invalid
ones~\citep{janota2025experiments}. The projects differ primarily in how their
search spaces are organized. ETP studies a fixed finite graph of implications
between individual laws, whereas \project organizes instances by finite
semigroup and asks whether the semigroup's full equational theory has a finite
basis. This requires producing a candidate basis and certifying its
completeness, or proving that no finite basis exists.
COPRA~\citep{copra} runs a language model inside a backtracking proof search
and retrieves lemmas from a library, and LEGO-Prover~\citep{legoprover} grows
its own library of verified lemmas; our family proofs and shared library play
the role of such a library at the scale of a catalogue. Studies of language
models as proof checkers find their verdicts sensitive to the prompt and to the
form of the proof~\citep{naik2026frontier,mazdarani2026answerkey}. Our referee
agent is therefore not trusted to verify: it triages work and settles disputes,
but a class is accepted only when the Lean kernel re-checks its proof in the
whole-corpus audit.

Our contributions are the following. (i)~Lean proofs, rebuilt from source in a
single audit, that each semigroup of order at most 6 has a finite basis or, for
four semigroups of order 6, none; to our knowledge this is the first
machine-checked account of a classification that so far existed only across
the literature~\citep{Lee2011,leezhang2015,survey}. (ii)~The complete inclusion order of the
\num{505} varieties generated by the finitely based semigroups of order 6,
determined by finite evaluation of their certified bases, and \num{289} hard
unit-equality problems encountered during proof search with
Vampire~\citep{vampire}, prepared for the TPTP library~\citep{tptp}.
(iii)~A description of the agent workflow and its acceptance procedure, with
measured effort, failure rates and the case where the agents stalled.
\section{Preliminaries}
\begin{table}
  \caption{Number of distinct semigroups per order, and how many of them are nilpotent.}\label{tab:counts}
  \begin{center}
    \begin{tabular}[c]{rrrr}
      \toprule
      Order & Semigroups & Nilpotent & $3$-nilpotent \\
      \midrule
      \num{1} & \num{1} & \num{1} & \num{0} \\
      \num{2} & \num{4} & \num{1} & \num{0} \\
      \num{3} & \num{18} & \num{2} & \num{1} \\
      \num{4} & \num{126} & \num{10} & \num{8} \\
      \num{5} & \num{1160} & \num{93} & \num{84} \\
      \num{6} & \num{15973} & \num{2813} & \num{2660} \\
      \num{7} & \num{836021} & \num{616830} & \num{609797} \\
      \bottomrule
    \end{tabular}
  \end{center}
\end{table}

A \emph{semigroup} is a structure with an associative binary operation. In the
remainder of the paper we denote the operation by $*$ and assume that it
operates on the elements $1,\dots,n$. The number of elements that the semigroup
operates on, $n$, is called the \emph{order}  of the semigroup. Whenever
convenient we omit the $*$; for instance, $x*x=x$ is abbreviated as $xx=x$ or $x^2=x$.

A finite semigroup is represented explicitly by its multiplication table.
Table~\ref{tab:counts} gives the number of \emph{distinct} semigroups,  where
\emph{distinct} means that isomorphic and anti-isomorphic semigroups are counted
only once  (we ignore element relabelling and transpositions of the multiplication table).

Fix a countably infinite set $\mathcal{X}$ of variables and let $\mathcal{X}^{+}$
be the nonempty words over $\mathcal{X}$ constructed by the operator $*$.
An \emph{identity} is an equality $u\approx v$ with
$u,v\in\mathcal{X}^{+}$. A semigroup $S$ \emph{satisfies} $u\approx v$, written
$S\models u\approx v$, if $\varphi(u)=\varphi(v)$ for every homomorphism
$\varphi\colon\mathcal{X}^{+}\to S$, that is, for every map $\mathcal{X}\to S$
extended multiplicatively to words. The \emph{equational theory} of $S$ is the set
$\mathrm{Id}(S)=\{\,u\approx v\mid S\models u\approx v\,\}$ of all identities
that hold in $S$.
A set of identities $\Sigma\subseteq\mathrm{Id}(S)$ is an \emph{identity basis}
of $S$ if every identity in $\mathrm{Id}(S)$ is derivable from $\Sigma$ by
Birkhoff's rules of equational deduction~\citep{birkhoff1935}; equivalently, if
every semigroup that satisfies $\Sigma$ satisfies all of $\mathrm{Id}(S)$.
Semigroup $S$ is \emph{finitely based} if it has a finite identity basis and
\emph{nonfinitely based} otherwise.

A semigroup has a \emph{zero element} $z$ if $\forall x.\ xz=zx=z$; additionally
it is \emph{nilpotent} if there is a
$k\in\mathbb{N}$ such that $\forall x_1,\dots,x_k.\ x_1\dots x_k=z$.
Nilpotent semigroups are a sixth of those of order 6 but three quarters of
those of order 7 (see \autoref{tab:counts}). They are
the easy case for our purposes, since a nilpotent semigroup of degree $k$
satisfies $x_1\dots x_k \approx y_1\dots y_k$ and is therefore finitely based
(intuitively, only short identities matter).

\begin{example}\label{ex:leftzero}
  Let $S_\ell$ be the semigroup on $\{1,2\}$ given by left projection,
  $x*y=x$; associativity holds because $(x*y)*z=x=x*(y*z)$. Under any
  assignment a word evaluates to the value of its leftmost variable, so
  $S_\ell\models u\approx v$ exactly when $u$ and $v$ begin with the same
  variable. The identity $xy\approx x$ is a basis: it rewrites any word to its
  first variable, after which the two sides of an identity of
  $\mathrm{Id}(S_\ell)$ coincide. Note that $S_\ell$ has no zero element, since
  $1*2=1$ and $2*1=2$, and is thus not nilpotent, yet it is still finitely
  based. \autoref{sec:proof} prints the Lean proof of this example.
\end{example}

A semigroup $T$ belongs to the \emph{variety generated by} $S$, denoted
$\mathsf{V}(S)$, exactly when it satisfies every identity of $S$~\citep{birkhoff1935}.
Any basis of $S$ therefore defines $\mathsf{V}(S)$.

\providecommand{\project}{\textsc{SemiBase}\xspace}

\section{The proof factory}\label{sec:factory}

\project produces its proofs with a small team of language-model agents
(\autoref{fig:factory} in the appendix). Three workers do the mathematics. A coordinator assigns
semigroups to them, allocates compute and runs independent rebuilds. A referee
certifies candidate bases, settles disputes between agents and decides what
enters the catalogue. The workers and the coordinator are Codex agents; the
referee is a Claude agent, so the reviewer and the reviewed come from different
model families. Humans at occasional intervals choose which semigroups to attack next and
restart stalled agents. Each agent appends records to a
shared log and can exchange files with other agents through per-agent mailboxes;
every record names the earlier records it builds on, so every proof can be traced back through the log to the human decision that started it.

\paragraph{Who did what.}
The pipeline that sealed the first \num{14989} classes is a set of scripts,
written by the agents at our direction in June and July, that runs without a
language model in the loop. It applies bases found by a bounded search over
short identities or taken from the literature. In this sense it is the
deterministic baseline: it settled 94\% of the classes and stopped. The
\num{984} classes it left open needed individual work: the agents proposed
candidate bases, searched for counter-models and wrote completeness proofs,
for a family of semigroups or for a single one. Published bases entered as
inputs where they existed, as for the neighbour of $L$ below. Humans wrote no
proofs: they chose targets, approved pushes and restarted agents,
with \num{440} prompts to the Codex agents on 72 days and \num{674} to the
referee on 61 days (\autoref{sec:effort}).

\paragraph{From a candidate to a proof.}
A candidate basis is first screened for counter-models: a semigroup that
satisfies the candidate but not some identity of the target refutes it.
A candidate that survives is written in Lean~4 as a theorem for the semigroup and for its \emph{transpose},
that is, anti-isomorphic copy. Showing that the
candidate identities hold on the table is a simple check; the main effort is in proving
that every identity of the table follows from the candidates. The
proofs that succeeded all have the same shape. A normal form for words is
fixed, together with a few rewriting rules derived from the basis, and every
word is shown to rewrite to its normal form. The normal form is then shown to
be determined by finitely many features of the word that the table can see,
such as the order in which letters first occur, the order in which they last
occur, and the parity of each letter. Semigroups that share a basis share this proof; only the finite
check on each table differs. We consider a given class \emph{sealed} once the referee has
rebuilt its proof from scratch together with the whole corpus, checked in the
kernel that the theorem's table is the catalogue table and that only the
standard axioms of Lean are used, and recorded the theorem names and file
hashes in a certificate. Both loops of \autoref{fig:factory} were busy: from
8 August to 6 September, \num{76} candidates were refuted before any Lean was
written, and \num{143} of \num{443} independent kernel checks sent a proof
back to its author; after two failed attempts, later three, a worker had to
stop and report to the referee.

\paragraph{The neighbour of $L$.}
The hardest class shows why the problem is difficult at all. It is the
semigroup $C_8$ of Section~19 of \citet{leezhang2015} (\num{3842} in
\textsc{Smallsemi}), and it differs from $L$, one of the four nonfinitely based
semigroups of order~6~\citep{Lee2012}, in exactly two entries of its
multiplication table. That difference is enough to cut off the infinite family
of identities that makes $L$ nonfinitely based~\citep{zhangluo2011}. The class
is finitely based, with the basis of their Proposition~19.1. The first candidate, found by
bounded search and approved by the referee, had nine identities in at most
three variables. A worker refuted it with a semigroup of order 20, larger than
any the screens examine, that satisfies all nine but not
$xykxty\approx xxkxtyyx$, an identity of the class. The agents switched to the
published basis, but its completeness proof did not close for another week,
until a structural analysis by the referee supplied a test for equality of
words in the class and a normal form. With these, one agent completed the proof
overnight, reusing much of the earlier attempt; it was the last proof of the
campaign (\autoref{sec:c8}).

\paragraph{Where the effort went.}
Two decisions saved the most work: sharing one basis across a family of
semigroups, and looking for counter-models before writing Lean. Two mistakes
cost the most: searching for bases with a word-length bound that was too small,
and checking classes one batch at a time instead of all together.
\section{Results}\label{sec:results}

\paragraph{Coverage.}
All \nLowOrder semigroups of order at most 5 and all \nSix semigroups of
order 6 are sealed. The four semigroups of order 6 known to be nonfinitely
based, $L$, $B_2^1$, $A_2^g$ and $A_2^1$~\citep{Lee2012} (numbers \num{3843},
\num{8564}, \num{8878} and \num{13747} in \textsc{Smallsemi}~\citep{smallsemi}),
carry Lean proofs that no finite basis exists; every other sealed class carries
a certified basis. The four proofs formalize the published
arguments~\citep{perkins1969,trahtman1987,sapir1988,zhangluo2011,leezhang2015}:
for every $k$, an identity of the semigroup does not follow from identities
in at most $k$ variables, because a property of words separating its two sides
survives every derivation step with them (\autoref{sec:nfb}).

\subsection{Proof size}\label{sec:proofsize}

The certificates differ in size by four orders of magnitude. Most classes are
covered by a family proof: one basis and one proof shared by every semigroup in
the family, with only the check on each multiplication table repeated. The
remaining \num{445} semigroups of order 6 have a generated proof of their own.
The common step is a transfer from a semigroup certified earlier: if
a semigroup $S$ with basis $B$ embeds in a direct power of the target $T$, or
is a quotient of a subsemigroup of $T$, then every identity of $T$ holds in $S$
and so follows from $B$; a check on the table that $T$ satisfies $B$ then makes
$B$ a basis for $T$. Others combine the bases of the factors of a subdirect
product, adapt the basis of a semigroup that $T$ inflates, or enumerate the
identities of a nilpotent $T$. \autoref{fig:sizes} (left)
shows the sizes of these generated proofs, all files of a proof counted
together. The median proof has a few hundred lines. Nine exceed
\num{100000} lines and together hold close to half of all generated proof text;
the largest runs to \num{333557}. Family proofs,
which an agent writes as ordinary Lean, are shorter, none above \num{10000}
lines; \autoref{tab:lengths} in the appendix gives the quantiles by kind of
file. Lines of Lean measure the size of a certificate, not the effort of
finding it: the largest proofs are generated mechanically, while the
\num{2918} new lines for the neighbour of $L$ took four weeks. The mathematics
sits in the family proofs (median 170 lines) and the shared library.

\begin{figure}[t]
  \centering
  \includegraphics[width=0.49\linewidth]{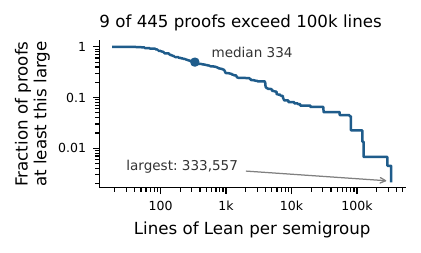}\hfill
  \includegraphics[width=0.49\linewidth]{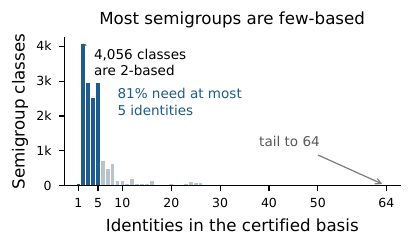}
  \caption{Left: share of the \num{445} generated proofs at least as long as a
  given number of lines, all files of a proof counted together (log axes).
  Right: the \num{15524} classes sealed by family proofs, by the size of their
  basis as certified (22 of all 536 certified bases repeat an identity up to
  renaming), before the reduction of \autoref{sec:basissize}; at most five
  identities highlighted.}\label{fig:sizes}
\end{figure}

\subsection{Basis size and the order of varieties}\label{sec:basissize}

\autoref{fig:sizes} (right) shows the sizes of the bases certified for the
\num{15524} classes sealed by family proofs. These bases fall into \num{486}
syntactically distinct lists of between \num{1} and \num{64} identities,
median \num{4}. Over four fifths of the
classes have a basis of at most five identities, \num{4056} have a basis of two
and \num{53} a basis of one; at the other end, \num{74} classes were sealed
with \num{26} identities and a few with up to \num{64}.

The certified bases need not be minimal, and distinct bases may define the
same variety. Vampire removes identities that follow from
the others, shrinking the \num{536} certified bases from
\num{7919} identities to \num{1781}; in \num{509} of the reduced bases every
remaining identity is shown independent of the others by a counter-model,
although an irredundant basis need not have minimum size. For
the \num{143380} pairs of bases, a semigroup of the order-6 catalogue
separates all but \num{31}, and Vampire proves each of those \num{31} pairs
equivalent. The bases therefore define exactly \num{505} distinct varieties.

Finite evaluation completes the inclusion order
(\autoref{sec:varietyorder}): \num{13325} strict inclusions,
\num{1317} covering pairs and \num{87} maximal varieties. The four directions
left undecided by Vampire all hold. The \num{289} hardest proof-search problems
are prepared for the TPTP library~\citep{tptp}. The Vampire proofs and table
evaluations have not been formalized in Lean.

\section{Summary and Conclusions}

We have introduced the project \project, which catalogues and formally
certifies finite identity bases of small semigroups. The finite basis problem,
originally proposed by \citet{tarski1968}, has a longstanding tradition in
universal algebra and is undecidable in general; the semigroup case remains
open. Identity bases also matter because they can be seen as syntactic
counterparts of varieties.

We show that the problem benefits from being attacked \emph{en masse} rather
than one semigroup at a time, because some semigroups share bases and some of
the proof techniques recur. At this stage of the project, all \nLowOrder
semigroups of order at most 5 and all \nSix of order 6 carry a Lean
certificate---either a finite basis or a proof that none exists.
The Lean development, the certificate of the final audit and scripts that check
both are archived on Zenodo~\citep{semibase2026zenodo}.

Agents produced the proofs, a referee and a whole-corpus audit made them
trustworthy, and humans chose targets and approved pushes but wrote no proofs.
All the proofs are kernel-checked in Lean~4.
The catalogue is itself an object of study: Vampire has reduced its bases,
finite evaluation has completed the inclusion order of the \num{505} varieties,
and \num{289} hard unit-equality problems provide benchmarks for automated
provers.
The hardness of the proofs is strongly long-tailed: many of the proofs are
almost routine, but a huge amount of resources must be dedicated to the last
remaining few. This also suggests where to look for interesting proof
techniques and problems; eventually we are likely to hit a case where the
agents fail.

The next target is order 7, with \num{836021} semigroups, \num{219191} of them
not nilpotent, an order of magnitude more than order 6. Order 6 took three
months, about \num{2260} agent sessions and \num{1500} core-hours of Lean
builds, and its cost sat in the tail: \num{15583} of its \num{15973} classes
were sealed by 14 August, the remaining \num{390} took 23 more days, and the
last needed a structural analysis before a proof could be written. Order 7
therefore needs the techniques used so far streamlined so that the tail, not
only the bulk, is discharged automatically. The workflow carries over to any
catalogue whose claims can be checked on each instance and stated in a proof
assistant.

\subsubsection*{Acknowledgements}
The research was supported by the European Union under the project
\textsf{ROBOPROX} (reg.\ no.\ CZ.02.01.01/00/22\_008/0004590) and is part of the
\textsf{RICAIP} project that has received funding from the European Union's
Horizon~2020 research and innovation programme under grant agreement No~857306.
Computational resources were supported by the Ministry of Education, Youth and Sports of the Czech Republic through the e-INFRA CZ (ID:90254).

\bibliographystyle{plainnat}

\clearpage
\appendix
\providecommand{\project}{\textsc{SemiBase}\xspace}

\section{The proof factory, step by step}\label{sec:factoryfig}

\autoref{fig:factory} shows the steps from a multiplication table to a sealed
class, arranged by who performs them; \autoref{sec:factory} describes them.

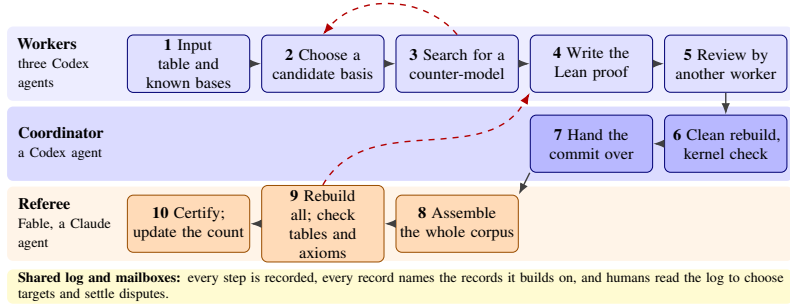
\begin{figure}[H]
\centering
\adjustbox{max width=\linewidth, max height=41mm}{%
\begin{tikzpicture}[x=1mm, y=1mm, font=\footnotesize,
  stp/.style={draw, rounded corners=3pt, align=center, inner sep=2pt, text width=20mm, minimum height=10mm, line width=0.5pt},
  wk/.style={stp, fill=blue!12, draw=blue!45!black},
  cd/.style={stp, fill=blue!28, draw=blue!45!black},
  sv/.style={stp, fill=orange!28, draw=orange!60!black},
  lbl/.style={anchor=west, align=left, font=\footnotesize\bfseries, text width=19mm, inner sep=1pt},
  arr/.style={-{Latex[length=2mm]}, line width=0.7pt, draw=black!75},
  loop/.style={-{Latex[length=2mm]}, line width=0.7pt, dashed, draw=red!70!black},
]
  \fill[blue!7, rounded corners=3pt] (0,-6.5) rectangle (137.5,6.5);
  \fill[blue!14, rounded corners=3pt] (0,-20.5) rectangle (137.5,-7.5);
  \fill[orange!9, rounded corners=3pt] (0,-34.5) rectangle (137.5,-21.5);
  \fill[yellow!22, rounded corners=3pt] (0,-42) rectangle (137.5,-36);
  \node[lbl] at (1.5,0) {Workers\\ {\scriptsize\mdseries three Codex agents}};
  \node[lbl] at (1.5,-14) {Coordinator\\ {\scriptsize\mdseries a Codex agent}};
  \node[lbl] at (1.5,-28) {Referee\\ {\scriptsize\mdseries Fable, a Claude agent}};
  \node[anchor=west, font=\scriptsize, inner sep=1pt, text width=134mm] at (1.5,-39)
    {\textbf{Shared log and mailboxes:} every step is recorded, every record names the records it builds on, and humans read the log to choose targets and settle disputes.};
  \node[wk] (s1) at (31.75,0)  {\textbf{1}~Input\\ table and known bases};
  \node[wk] (s2) at (55.25,0)  {\textbf{2}~Choose a\\ candidate basis};
  \node[wk] (s3) at (78.75,0)  {\textbf{3}~Search for a\\ counter-model};
  \node[wk] (s4) at (102.25,0) {\textbf{4}~Write the\\ Lean proof};
  \node[wk] (s5) at (125.75,0) {\textbf{5}~Review by\\ another worker};
  \node[cd] (s6) at (125.75,-14) {\textbf{6}~Clean rebuild,\\ kernel check};
  \node[cd] (s7) at (102.25,-14) {\textbf{7}~Hand the\\ commit over};
  \node[sv] (s8)  at (78.75,-28) {\textbf{8}~Assemble\\ the whole corpus};
  \node[sv] (s9)  at (55.25,-28) {\textbf{9}~Rebuild all; check\\ tables and axioms};
  \node[sv] (s10) at (31.75,-28) {\textbf{10}~Certify;\\ update the count};
  \foreach \a/\b in {s1/s2, s2/s3, s3/s4, s4/s5, s6/s7, s8/s9, s9/s10} \draw[arr] (\a) -- (\b);
  \draw[arr] (s5) -- (s6);
  \draw[arr] (s7.south west) -- (s8.north east);
  \draw[loop] (s3.north) to[bend right=50] (s2.north);
  \draw[loop] (s9.north) to[out=60, in=-120] (s4.south west);
\end{tikzpicture}}
\caption{The steps from a multiplication table to a sealed class, arranged by who performs them (blue: Codex agents; orange: the referee). The dashed red arrows return work to an earlier step: a refuted candidate goes back to step~2, a proof that fails the audit back to step~4.}
\label{fig:factory}
\end{figure}

\section{How the Lean proofs are organised}\label{sec:anatomy}

\autoref{fig:anatomy} shows how the Lean code is layered and what a single
proof module contains. The acceptance repository, which adds the audit and
wrapper modules to the proof repository counted in Section~\ref{sec:proofsize},
holds \num{9329} files and \num{6.94} million lines, on Lean~4.28 without
Mathlib. The shared library is small; nearly all of the volume is proofs and
audits about individual semigroups.

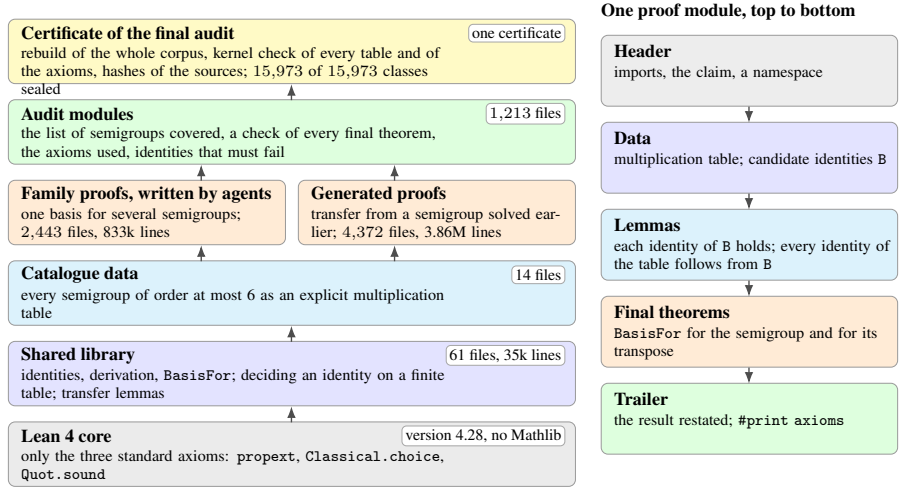
\begin{figure}[!htb]
\centering
\adjustbox{max width=\linewidth, max height=64mm}{%
\begin{tikzpicture}[x=1mm, y=1mm, font=\scriptsize,
  band/.style={draw=black!40, rounded corners=3pt, anchor=south west, minimum width=88mm, minimum height=10mm, inner sep=0pt, line width=0.4pt},
  half/.style={band, minimum width=43mm},
  ttl/.style={anchor=west, font=\footnotesize\bfseries, inner sep=0pt},
  sub/.style={anchor=north west, font=\scriptsize, inner sep=0pt, text width=66mm},
  badge/.style={anchor=east, font=\scriptsize, inner sep=1.5pt, fill=white, rounded corners=2pt, draw=black!35},
  part/.style={draw=black!40, rounded corners=3pt, anchor=south west, minimum width=46mm, minimum height=11mm, inner sep=0pt, line width=0.4pt},
  psub/.style={anchor=north west, font=\scriptsize, inner sep=0pt, text width=43mm},
  arr/.style={-{Latex[length=1.8mm]}, line width=0.6pt, draw=black!70},
]
  \node[band, fill=gray!16] at (0,0) {};
  \node[ttl] at (2,7.9) {Lean 4 core};
  \node[sub] at (2,5.8) {only the three standard axioms: \texttt{propext}, \texttt{Classical.choice}, \texttt{Quot.sound}};
  \node[badge] at (86.5,7.9) {version 4.28, no Mathlib};

  \node[band, fill=blue!11] at (0,12.5) {};
  \node[ttl] at (2,20.4) {Shared library};
  \node[sub] at (2,18.3) {identities, derivation, \texttt{BasisFor}; deciding an identity on a finite table; transfer lemmas};
  \node[badge] at (86.5,20.4) {61 files, 35k lines};

  \node[band, fill=cyan!12] at (0,25) {};
  \node[ttl] at (2,32.9) {Catalogue data};
  \node[sub] at (2,30.8) {every semigroup of order at most 6 as an explicit multiplication table};
  \node[badge] at (86.5,32.9) {14 files};

  \node[half, fill=orange!16] at (0,37.5) {};
  \node[ttl] at (2,45.4) {Family proofs, written by agents};
  \node[sub, text width=39mm] at (2,43.3) {one basis for several semigroups; \num{2443} files, 833k lines};
  \node[half, fill=orange!16] at (45,37.5) {};
  \node[ttl] at (47,45.4) {Generated proofs};
  \node[sub, text width=39mm] at (47,43.3) {transfer from a semigroup solved earlier; \num{4372} files, 3.86M lines};

  \node[band, fill=green!13] at (0,50) {};
  \node[ttl] at (2,57.9) {Audit modules};
  \node[sub] at (2,55.8) {the list of semigroups covered, a check of every final theorem, the axioms used, identities that must fail};
  \node[badge] at (86.5,57.9) {\num{1213} files};

  \node[band, fill=yellow!28] at (0,62.5) {};
  \node[ttl] at (2,70.4) {Certificate of the final audit};
  \node[sub] at (2,68.3) {rebuild of the whole corpus, kernel check of every table and of the axioms, hashes of the sources; \nSealedSix of \nSix classes sealed};
  \node[badge] at (86.5,70.4) {one certificate};

  \draw[arr] (44,10) -- (44,12.5);
  \draw[arr] (44,22.5) -- (44,25);
  \draw[arr] (30,35) -- (30,37.5);
  \draw[arr] (60,35) -- (60,37.5);
  \draw[arr] (30,47.5) -- (30,50);
  \draw[arr] (60,47.5) -- (60,50);
  \draw[arr] (44,60) -- (44,62.5);

  \node[anchor=south west, font=\footnotesize\bfseries, inner sep=0pt] at (92,72.5) {One proof module, top to bottom};
  \node[part, fill=gray!14] at (92,59) {};
  \node[ttl] at (94,67.6) {Header};
  \node[psub] at (94,65.4) {imports, the claim, a namespace};
  \node[part, fill=blue!11] at (92,45.5) {};
  \node[ttl] at (94,54.1) {Data};
  \node[psub] at (94,51.9) {multiplication table; candidate identities \texttt{B}};
  \node[part, fill=cyan!12] at (92,32) {};
  \node[ttl] at (94,40.6) {Lemmas};
  \node[psub] at (94,38.4) {each identity of \texttt{B} holds; every identity of the table follows from \texttt{B}};
  \node[part, fill=orange!16] at (92,18.5) {};
  \node[ttl] at (94,27.1) {Final theorems};
  \node[psub] at (94,24.9) {\texttt{BasisFor} for the semigroup and for its transpose};
  \node[part, fill=green!13] at (92,5) {};
  \node[ttl] at (94,13.6) {Trailer};
  \node[psub] at (94,11.4) {the result restated; \texttt{\#print axioms}};
  \draw[arr] (115,59) -- (115,56.5);
  \draw[arr] (115,45.5) -- (115,43);
  \draw[arr] (115,32) -- (115,29.5);
  \draw[arr] (115,18.5) -- (115,16);
\end{tikzpicture}}
\caption{Left: the layers of the Lean code, each built on the one below. Right: the parts of a single proof module.}
\label{fig:anatomy}
\end{figure}

\paragraph{What a proof module contains.}
Every module, whether it treats one semigroup or a family of semigroups with a
common basis, has the same shape. The header imports the shared library and
the family's own lemma files, states the claim in a comment and opens a
namespace named after the semigroup. The data are the multiplication table and
the candidate identities. The lemmas show that each candidate identity holds
on the table, which Lean checks by computation, and that every identity the
table satisfies follows from the candidates; this second half is proved through
normal forms or by decomposing the semigroup into a subdirect product of
smaller semigroups that were treated before. The final theorems state
\texttt{BasisFor} for the semigroup and for its transpose. A short trailer
restates the result and prints the axioms used. Generated modules differ only
in the middle: they show that the target has the same identities as a
semigroup solved earlier and reuse its basis; their heavy computations are
split into several files.

\paragraph{Shared files.}
The core library is imported by everything. Its central file,
\texttt{Equational.lean}, has 91 lines; it defines identities, satisfaction,
models, derivation and \texttt{BasisFor}, and proves that derivation is sound.
Each family keeps its own lemma files next to the proofs of its members, and
each module imports only what it needs. A batch is therefore built from a fixed
commit, and its audit module lists exactly which semigroups it covers.

\paragraph{The campaign.}
The campaign covers the \num{15973} semigroups of order~6. A scripted pipeline
first sealed \num{14989} of them and left \num{984} open. The agents then took
the remainder in \num{102} batches and sealed \num{943} more: \num{29} batches
of ten or more semigroups with a common strategy, and \num{73} small batches
down to single hard cases. The last \num{41} were delivered in the final week
and certified together by the whole-corpus audit, after which all \nSix
classes are sealed. Under the batch process used until 5 September, a batch
with a single semigroup took 37~minutes from submission to sealing and an
independent verification job 2 to 20 minutes. The largest single build, the
\num{333557}-line proof of one semigroup, took more than a day on one core.

\paragraph{Proof lengths.}
Section~\ref{sec:proofsize} plots the size of the generated proofs per
semigroup. \autoref{tab:lengths} gives the quantiles at the certified commit,
per semigroup and by kind of file. The distribution has a heavy tail. Among
the \num{445} semigroups with a generated proof of their own, the median proof
has 334 lines but the mean is \num{7227}. Nine semigroups need more than
\num{100000} lines each and together hold close to half of all generated proof
text; the largest, at \num{333557} lines, is the proof whose build took a day.
Family proofs, written by agents, are much shorter (median 170 lines, none
above \num{10000}), because a shared basis and a subdirect decomposition
replace the direct transfer that a generated proof carries out. The typical
proof can be read by a person; the largest cannot, a growth in certificate size
that recalls the busy-beaver problem~\citep{Rado1962}. The last classes to be
sealed belong to this tail.

\begin{table}[H]
\centering\scriptsize
\adjustbox{max width=\linewidth}{%
\begin{tabular}{@{}lrrrrrrr@{}}
\toprule
Kind of unit (lines of Lean) & $n$ & median & p75 & p90 & p99 & max & total \\
\midrule
Generated proof of one semigroup (module and its parts) & 445 & 334 & \num{1466} & \num{7398} & \num{126350} & \num{333557} & 3.22M \\
Family proof file, written by an agent & \num{2443} & 170 & 399 & 842 & \num{2478} & \num{9568} & 833k \\
Audit module & \num{1213} & 32 & 63 & 302 & \num{1669} & \num{3626} & 230k \\
Shared library file & 61 & 207 & 382 & 538 & \num{18584} & \num{18584} & 35k \\
\bottomrule
\end{tabular}}\\[4pt]
\adjustbox{max width=\linewidth}{%
\begin{tabular}{@{}l*{8}{r}@{}}
\toprule
Generated proofs by size (lines) & $<100$ & 100--300 & 300--1k & 1k--3k & 3k--10k & 10k--30k & 30k--100k & $\ge 100$k \\
\midrule
semigroups & 75 & 142 & 94 & 40 & 59 & 7 & 19 & 9 \\
\bottomrule
\end{tabular}}
\caption{Proof lengths at the certified commit of the final audit: quantiles over units of each kind (top), and the number of generated per-semigroup proofs in each size class (bottom).}
\label{tab:lengths}
\end{table}

\section{Effort and resources}\label{sec:effort}

\autoref{tab:effort} collects what the campaign cost, measured from the
repository history, the agents' session logs, the shared log and the
accounting of the two clusters that ran the Lean builds. Entries marked
``about'' are rounded. The shared log was introduced for the finishing
campaign on 8 August, so the counts taken from it cover the last month only.

\begin{table}[H]
\centering\small
\begin{tabular}{@{}>{\raggedright\arraybackslash}p{0.34\linewidth}>{\raggedright\arraybackslash}p{0.62\linewidth}@{}}
\toprule
Duration & 9 June to 6 September 2026: 90 days, with commits on 84 of them \\
Commits & \num{4923} on all branches, \num{3588} of them in the history of the certified commit \\
Workers and coordinator & Codex agents (Codex CLI, GPT-5-family models): about \num{2260} sessions, of which about 660 in June, \num{1510} in July and 90 later; over 100\,GB of session logs \\
Referee & a Claude agent in Claude Code: \num{43236} model turns and 41.9 million output tokens up to 6 September \\
Human prompts & \num{440} to the Codex agents on 72 days, \num{674} to the referee on 61 days \\
Shared log, 8 August to 6 September & \num{4046} records: \num{1609} from the coordinator, \num{1242} from the workers, 385 from the referee, the rest job records \\
Candidates refuted by a counter-model before any Lean was written & 76 \\
Independent kernel checks & 443 jobs: 291 passed, 143 failed and sent the proof back to its author, 9 ended otherwise \\
Referee records that reject, withdraw or hold a proposal & about 106 of 385 \\
Lean builds on the WMI cluster & \num{18222} jobs, about \num{1500} core-hours \\
Largest single build & S6\_12824, a proof of \num{333557} lines: more than 24 hours on one core \\
Final audit & about \num{7200} modules; seven runs over two days, the last, incremental, 4\,h\,6\,min on 4 cores \\
Independent re-certification & a fresh clone rebuilt with 96 parallel jobs in 2\,h: all \num{15973} classes, no failures \\
\bottomrule
\end{tabular}
\caption{Effort and resources of the order-6 campaign.}\label{tab:effort}
\end{table}

\paragraph{The long tail.}
\autoref{tab:census} shows the number of sealed classes over the finishing
campaign. The scripted pipeline left \num{984} classes open; by 14 August only
\num{390} remained, and those took another 23 days. The last \num{41} were
certified together by the final audit.

\begin{table}[H]
\centering\small
\begin{tabular}{@{}lr@{\hspace{2em}}lr@{\hspace{2em}}lr@{}}
\toprule
Date & Sealed & Date & Sealed & Date & Sealed \\
\midrule
scripted stage & \num{14989} & 28 August & \num{15776} & 2 September & \num{15928} \\
14 August & \num{15583} & 29 August & \num{15856} & 3 September & \num{15930} \\
25 August & \num{15649} & 30 August & \num{15886} & 4 September & \num{15932} \\
26 August & \num{15711} & 31 August & \num{15894} & 6 September & \num{15973} \\
27 August & \num{15720} & 1 September & \num{15921} & & \\
\bottomrule
\end{tabular}
\caption{Classes of order 6 sealed, at the last census of each day on which
the census changed. The final audit of 6 September certified all \nSix
together.}\label{tab:census}
\end{table}

\paragraph{Nilpotent classes.}
The records do not let us split the compute between nilpotent and other
classes. Nilpotent classes are easy mathematically but not always cheap: the
second-largest generated proof, \num{291860} lines, certifies the nilpotent
class S6\_2582, whose basis has \num{3513} identities.

\section{Worked example: the neighbour of \texorpdfstring{$L$}{L}}\label{sec:c8}

\paragraph{The semigroup.}
The class S6\_3842, the semigroup $C_8$ of Section~19 of \citet{leezhang2015},
has the elements $0, ex, x, ed, d, e$ and the table below,
row times column. The idempotents $d$ and $e$ generate $\{0,ed,d,e\}$ with
$de=0$, and $\{0,x,ex\}$ is a null ideal with $dx=x$, $xe=x$ and $xd=0$. The
two shaded entries are the whole difference from $L$ (S6\_3843), in which
$d\cdot e$ and $ed\cdot e$ are not zero.

\begin{center}\small
\begin{tabular}{r|cccccc}
$\cdot$ & $0$ & $ex$ & $x$ & $ed$ & $d$ & $e$\\
\midrule
$0$  & $0$ & $0$  & $0$ & $0$  & $0$  & $0$\\
$ex$ & $0$ & $0$  & $0$ & $0$  & $0$  & $ex$\\
$x$  & $0$ & $0$  & $0$ & $0$  & $0$  & $x$\\
$ed$ & $0$ & $0$  & $ex$& $0$  & $ed$ & \colorbox{orange!25}{$0$}\\
$d$  & $0$ & $0$  & $x$ & $0$  & $d$  & \colorbox{orange!25}{$0$}\\
$e$  & $0$ & $ex$ & $ex$& $ed$ & $ed$ & $e$\\
\end{tabular}
\end{center}

\paragraph{Why it is finitely based.}
What separates the two semigroups is the product $x\,y^2\,x$: under
$x\mapsto e$, $y\mapsto d$ it is $ede\neq0$ in $L$ but $0$ in S6\_3842,
because $ed\cdot e=0$. As a result S6\_3842 satisfies identities that fail in
$L$, such as $xyxy\approx xyyx$ and the exchange of adjacent squares
$x\,y^2z^2t^2\,x\approx x\,z^2y^2t^2\,x$. The family
$x\,y_1^2\cdots y_n^2\,x\approx x\,y_n^2\cdots y_1^2\,x$, which no finite set of
identities of $L$ derives (\autoref{sec:nfb}), holds in both semigroups, but
in S6\_3842 it follows from a basis whose identities have at most five
variables.

\paragraph{The basis.}
The certified basis is Proposition~19.1 of \citet{leezhang2015}: the twelve
schemata below, in which each of $H$, $K$, $T$ is independently either absent
or one further variable, expand to 38 identities.

\begin{center}\small
\begin{tabular}{@{}l@{\hspace{3em}}l@{}}
$xxx\approx xx$ & $xHyyx\approx yyxHyyx$ \\
$xxyx\approx xyx$ & $xHyyxKx\approx xHyyxKyyx$ \\
$xyxx\approx xyx$ & $xHyKxy\approx xHyKyx$ \\
$xyxzx\approx xzxyx$ & $xHyKxy\approx yHxKxy$ \\
$xHyKxTy\approx xHxKxTyyx$ & $xHyyxKy\approx xHyyxKx$ \\
& $yHxKyyx\approx xHxKyyx$ \\
& $xHyKyyx\approx xHxKyyx$ \\
\end{tabular}
\end{center}

The first candidate consisted of $xx\approx xxx$, $xxyx\approx xyx$,
$xyxx\approx xyx$, $xyxy\approx xyyx$, $xyxy\approx yxxy$,
$xyxzx\approx xzxyx$, $xyxzy\approx xyzyx$, $xyxzy\approx xzyyx$ and
$xyxzy\approx xyzxy$, all in at most three variables. It does not derive
$xykxty\approx xxkxtyyx$, the instance of the fifth schema with $H$ absent: a
worker found a semigroup of order 20 that satisfies the nine identities but not
this one. Replacing the ninth identity by the five-variable instance
$xhykxty\approx xhxkxtyyx$ gives a basis equivalent to the published one; we
checked the equivalence by explicit derivations in both directions, outside
Lean.

\paragraph{The decisive observation.}
Two words are equal in S6\_3842 exactly when they agree on a few combinatorial
features. Call a letter \emph{simple} if it occurs once, and cut a word into
\emph{components} wherever the letters before the cut and after it are
disjoint. Words $u$ and $v$ are equal in the semigroup if and only if they
have the same letters, the same simple letters and the same sequence of
component contents, and, around each simple letter $t$, written $u=p\,t\,s$,
the same contents of the trailing components of $p$ that avoid the letters of
$s$, and agree on whether $p$ and $s$ share a letter. A normal form follows:
each component is rewritten to a canonical word determined by these features.
Completeness then reduces to two facts, that the basis rewrites every word to
its normal form and that words with the same features have the same normal
form. The criterion was checked by computation on all words in up to four
variables of bounded length before it was given to the agent as the plan of
the proof.

\paragraph{How the proof was found.}
\begin{center}\small
\begin{tabular}{@{}>{\raggedright\arraybackslash}p{0.17\linewidth}>{\raggedright\arraybackslash}p{0.77\linewidth}@{}}
\toprule
8 August & The class enters the finishing campaign; a worker records its structure. \\
29 August & A worker's semigroup of order 20 refutes the nine-identity candidate; the referee withdraws it the same day and makes the published basis the only route. \\
29 August to 5 September & Parts of Lee and Zhang's argument are proved in Lean, but the completeness proof does not close. \\
5 September & The referee's structural analysis gives the criterion and the normal form above. \\
6 September, 01:10 & A dedicated agent delivers the proof: 25 new modules with \num{2918} lines and 169 theorems, reusing 63 modules (\num{14497} lines) of the earlier attempt. \\
6 September & The final audit certifies the class together with all others. \\
\bottomrule
\end{tabular}
\end{center}
The certificate consists of two theorems, \texttt{BasisFor} for the table of
S6\_3842 with the 38 identities and for its transpose with the reversed
identities, which use only the three standard axioms. Over the four weeks the
shared log has 86 records about the class, from four agents.

\section{The four proofs that no finite basis exists}\label{sec:nfb}

In Lean a semigroup is nonfinitely based when no finite list of identities over
the variables $\mathbb{N}$ is a basis for it (\autoref{sec:proof}). The audit
accepts a theorem as the proof for one of the four classes only if its
statement is exactly \texttt{NonfinitelyBased} of the catalogue table, with no
hypotheses, and if it uses only the three standard axioms.

All four proofs follow one scheme. They exhibit identities $u_n\approx v_n$ of
the semigroup, one for every $n$, and a property of words that separates
$u_n$ from $v_n$. For every bound $k$, once $n$ is large enough compared with
$k$, no derivation step that applies an identity in at most $k$ variables can
change the property, so $u_n\approx v_n$ does not follow from such
identities. The identities of a finite set have at most $k$ variables for some
$k$, so the set is not a basis. The proofs differ in the family and in the
property (\autoref{tab:nfb}).

\begin{table}[H]
\centering\small
\begin{tabular}{@{}>{\raggedright\arraybackslash}p{0.1\linewidth}>{\raggedright\arraybackslash}p{0.36\linewidth}>{\raggedright\arraybackslash}p{0.34\linewidth}>{\raggedright\arraybackslash}p{0.09\linewidth}@{}}
\toprule
Class & Identities $u_n\approx v_n$ & Property preserved by short steps & Lean \\
\midrule
$L$ \newline S6\_3843 &
$x\,y_1^2\cdots y_n^2\,x\approx x\,y_n^2\cdots y_1^2\,x$ &
words equal in $L$ to $u_n$ fall into two classes, with the squares in the
forward or in the reverse order, and no short step crosses between them &
6 files, \num{5337} lines \\
$B_2^1$ \newline S6\_8564 &
$x\,y_1\cdots y_n\,x\,y_n\cdots y_1\approx x\,y_n\cdots y_1\,x\,y_1\cdots y_n$ &
a pattern of the occurrences of the letters; the argument uses that
$xytyx$ and $xtyxy$ are isoterms, equal in $B_2^1$ only to themselves &
8 files, \num{12747} lines \\
$A_2^1$ \newline S6\_13747 &
$X\,y\,X'\,y\,X\approx X\,y\,X'\,y\,X\,y\,X'\,y\,X$, where $X=x_1\cdots x_n$
and $X'$ is its reverse &
the word $u_n$ itself: no short step changes it &
8 files, \num{10502} lines \\
$A_2^g$ \newline S6\_8878 &
$(x_1^2\cdots x_n^2)^2\approx(x_1^2\cdots x_n^2)^3$ &
the parity with which a word winds around the cycle
$x_1\to x_2\to\dots\to x_n\to x_1$ &
2 files, \num{2625} lines \\
\bottomrule
\end{tabular}
\caption{The four nonfinitely based semigroups of order 6. A short step applies
an identity whose number of variables is small compared with $n$. Line counts
exclude audit modules and the 404 lines of shared definitions.}\label{tab:nfb}
\end{table}

\paragraph{Sources.}
The families and arguments are those of \citet{zhangluo2011} for $L$, of
\citet{perkins1969} in the reconstruction of \citet{sapir1988} for $B_2^1$,
and of \citet{trahtman1987} and \citet{sapir1988} for $A_2^1$. For $L$ the
formalization reconstructs the proof of Lemma~2 of \citet{zhangluo2011}, on
which their reduction to connected identities rests. For $A_2^g$ the proof
characterizes the identities of the semigroup by a marked directed graph of
the word together with a parity, and the winding argument works on that
characterization. The four semigroups are those of the main theorem of
\citet{leezhang2015}. Where the catalogue lists a published table with its
elements in another order ($B_2^1$ and $A_2^g$), the proof checks the
relabelling and shows that both tables satisfy the same identities.

\section{A representative proof}\label{sec:proof}

Every certificate refers to the definitions below, from the shared library
(\texttt{Equational.lean} and \texttt{Nonfinite.lean}, abridged to the
definitions). A word is a nonempty list of variables, \texttt{G.eval} evaluates
a word in \texttt{G} under a valuation of its variables, and \texttt{Derives}
is equational deduction with substitution of words for variables.

\begin{lstlisting}
def SatisfiedBy (e : Identity α) (G : Semigroup S) : Prop :=
  ∀ valuation : α → S, G.eval valuation e.lhs = G.eval valuation e.rhs

def Models (G : Semigroup S) (basis : List (Identity α)) : Prop :=
  ∀ e, e ∈ basis → e.SatisfiedBy G

/-- Equational derivability for semigroup words over one infinite variable type. -/
inductive Derives (basis : List (Identity α)) : Word α → Word α → Prop
  | fromBasis {e : Identity α} : e ∈ basis → Derives basis e.lhs e.rhs
  | refl (u : Word α) : Derives basis u u
  | symm {u v : Word α} : Derives basis u v → Derives basis v u
  | trans {u v w : Word α} :
      Derives basis u v → Derives basis v w → Derives basis u w
  | prepend (p : Word α) {u v : Word α} :
      Derives basis u v → Derives basis (p ++ u) (p ++ v)
  | appendRight {u v : Word α} :
      Derives basis u v → (q : Word α) → Derives basis (u ++ q) (v ++ q)
  | subst {u v : Word α} :
      Derives basis u v → (σ : α → Word α) →
        Derives basis (u.bind σ) (v.bind σ)

/-- A finite list is a basis when it is sound and derives every valid identity. -/
def BasisFor (G : Semigroup S) (basis : List (Identity α)) : Prop :=
  Models G basis ∧
    ∀ e : Identity α, e.SatisfiedBy G → Derives basis e.lhs e.rhs

def FinitelyBasedOver (G : Semigroup S) (α : Type u) : Prop :=
  ∃ basis : List (Identity α), BasisFor G basis

/-- Failure of finite basability relative to a chosen variable type. -/
def NonfinitelyBasedOver (G : Semigroup S) (α : Type u) : Prop :=
  ¬FinitelyBasedOver G α

/-- The standard nonfinite-basis predicate, using `Nat` variables. -/
def NonfinitelyBased (G : Semigroup S) : Prop :=
  NonfinitelyBasedOver G Nat
\end{lstlisting}

The proof below is complete. It is the module \texttt{Examples/LeftZeroTwo.lean}
of the certified tree, for the left-zero semigroup $S_\ell$ of
Example~\ref{ex:leftzero}, up to its theorem for the semigroup itself; the
remaining 39 lines derive the theorem for the transpose from it. A family
proof of order 6 has the same parts at a larger scale: the table, the basis,
soundness by computation, a normal form reached by derivation (here the first
letter of a word, \texttt{leftZeroDerivesHead}), and completeness by a
valuation that tells different normal forms apart.

\begin{lstlisting}
import SemigroupBasis.FiniteTable
import SemigroupBasis.Opposite

namespace SemigroupBasis.Examples

open SemigroupBasis

def leftZeroTwo : FiniteTable where
  order := 2
  mul := fun a _ => a
  assoc := by intros; rfl

def xy : Word (Fin 2) := ⟨0, [1]⟩
def x : Word (Fin 2) := Word.singleton 0
def leftZeroLaw : Identity (Fin 2) := ⟨xy, x⟩

theorem leftZeroLaw_valid : leftZeroLaw.SatisfiedBy leftZeroTwo.semigroup :=
  leftZeroTwo.checkIdentity_sound leftZeroLaw (by decide)

def natXY : Word Nat := ⟨0, [1]⟩
def natX : Word Nat := Word.singleton 0
def leftZeroBasisLaw : Identity Nat := ⟨natXY, natX⟩
def leftZeroBasis : List (Identity Nat) := [leftZeroBasisLaw]

private def instantiateTwo (u v : Word Nat) : Nat → Word Nat
  | 0 => u
  | _ => v

theorem leftZeroTwo_eval (valuation : Nat → Fin 2) (w : Word Nat) :
    leftZeroTwo.semigroup.eval valuation w = valuation w.head := by
  cases w with
  | mk head tail =>
      induction tail with
      | nil => rfl
      | cons y ys ih =>
          simp only [Semigroup.eval, List.foldl_cons]
          exact ih

theorem leftZeroBasis_models : Models leftZeroTwo.semigroup leftZeroBasis := by
  intro e he
  simp only [leftZeroBasis, List.mem_singleton] at he
  subst e
  intro valuation
  rfl

theorem leftZeroDerivesHead (w : Word Nat) :
    Derives leftZeroBasis w (Word.singleton w.head) := by
  cases w with
  | mk head tail =>
      cases tail with
      | nil => exact Derives.refl _
      | cons next rest =>
          let suffix : Word Nat := ⟨next, rest⟩
          have hbase : Derives leftZeroBasis natXY natX :=
            Derives.fromBasis (e := leftZeroBasisLaw) <| by
              change leftZeroBasisLaw ∈ [leftZeroBasisLaw]
              exact List.Mem.head []
          have h := Derives.subst hbase
            (instantiateTwo (Word.singleton head) suffix)
          simpa [natXY, natX, instantiateTwo, Word.bind, Word.append,
            Word.singleton, suffix] using h

theorem leftZeroBasis_complete : BasisFor leftZeroTwo.semigroup leftZeroBasis := by
  refine ⟨leftZeroBasis_models, ?_⟩
  intro e hsatisfied
  have hheads : e.lhs.head = e.rhs.head := by
    apply Decidable.byContradiction
    intro hne
    let valuation : Nat → Fin 2 := fun z =>
      if z = e.lhs.head then 0 else 1
    have heval := hsatisfied valuation
    rw [leftZeroTwo_eval, leftZeroTwo_eval] at heval
    simp [valuation, Ne.symm hne] at heval
  exact Derives.trans (leftZeroDerivesHead e.lhs) <| by
    rw [hheads]
    exact Derives.symm (leftZeroDerivesHead e.rhs)
\end{lstlisting}

\section{Computing the inclusion order}\label{sec:varietyorder}

The following finite inclusion test follows from
\citet[Chapter~II, Lemmas~11.3 and~11.8]{burris1981}.

\begin{proposition}\label{prop:variety-inclusion}
Let $S,T$ be finite semigroups and let $B_T$ be a finite identity basis of $T$.
Then
\[
  \mathsf{V}(S)\subseteq\mathsf{V}(T)
  \quad\Longleftrightarrow\quad S\models B_T.
\]
Consequently, this inclusion is decidable by finite evaluation of $B_T$ on $S$.
\end{proposition}

\begin{proof}
  Since $V(S)$ is the smallest variety containing $S$, the inclusion
  $V(S)\subseteq V(T)$ holds if and only if $S\in V(T)$. As $B_T$ is an identity
  basis of $T$, this is equivalent to $S\models B_T$. Each identity in $B_T$
  involving $k$ variables can be checked on all $|S|^k$ assignments. Since $S$
  and $B_T$ are finite, this gives a terminating decision procedure.
\end{proof}

Evaluating each certified basis on a generator of each of the \num{505}
varieties confirms all \num{13321} inclusions proved by Vampire or
transitivity and adds exactly the four directions that remained undecided
after 12 hours of search each. The complete order has \num{13325} strict
inclusions, \num{1317} covering pairs and \num{87} maximal varieties.
Explicit equational derivations have since been found for three of the four
directions; one remains to be constructed. The inclusion itself is settled
by the finite check and the certified completeness of the bases.

The bases are certified in Lean. The finite table evaluations and Vampire's
equational derivations have not been formalized in Lean.
 
\end{document}